\documentclass[journal]{IEEEtran}

\makeindex             

\usepackage{amsthm}
\usepackage{epsfig} 
\usepackage{epstopdf}
\usepackage{amsmath}
\usepackage{amsthm}
\usepackage{comment}
\usepackage{url}
\usepackage{algorithm}
\usepackage{algorithmic}

\usepackage{amsfonts}
\usepackage{bm}
\usepackage{color}
\usepackage{cite}
\usepackage{stfloats}
\usepackage{chemarrow}
\usepackage{extarrows}

\usepackage{url}

\usepackage{setspace}

\DeclareGraphicsExtensions{.pdf,.png,.jpg,.eps,.ps}

\def\baa{\begin{align}}
\def\eaa{\end{align}}

\newcommand{\bsq}{\begin{subequations}}
\newcommand{\esq}{\end{subequations}}

\newcommand{\beq}{\begin{equation}}
\newcommand{\eeq}{\end{equation}}
\newcommand{\bq}{\begin{eqnarray}}
\newcommand{\eq}{\end{eqnarray}}
\newcommand{\bqn}{\begin{eqnarray*}}
	\newcommand{\eqn}{\end{eqnarray*}}
\newcommand{\bee}{\begin{enumerate}}
	\newcommand{\eee}{\end{enumerate}}
\newcommand{\bi}{\begin{itemize}}
	\newcommand{\ei}{\end{itemize}}

\newcommand{\diag}{\mathrm{diag}}

\newboolean{showcomments}
\setboolean{showcomments}{true}
\newcommand{\wang}[1]{\ifthenelse{\boolean{showcomments}}
	{ \textcolor[rgb]{1,0,1}{(ZW:  #1)}}{}}
\newcommand{\fliu}[1]{\ifthenelse{\boolean{showcomments}}
	{ \textcolor{red}{(FL:  #1)}}{}}
\newcommand{\zhao}[1]{\ifthenelse{\boolean{showcomments}}
	{ \textcolor{green}{(JP:  #1)}}{}}
\newcommand{\slow}[1]{\ifthenelse{\boolean{showcomments}}
	{ \textcolor{blue}{(SL:  #1)}}{}}

\theoremstyle{definition}
\newtheorem{theorem}{Theorem}
\newtheorem{lemma}[theorem]{Lemma}

\theoremstyle{definition}
\newtheorem{definition}{Definition}
\newtheorem{remark}{Remark}

\newtheorem{assumption}{\textit{Assumption}}

\ifCLASSOPTIONcaptionsoff
\usepackage[nomarkers]{endfloat}
\let\MYoriglatexcaption\caption
\renewcommand{\caption}[2][\relax]{\MYoriglatexcaption[#2]{#2}}
\fi

\allowdisplaybreaks

\begin{document}
\graphicspath{{Paper_Fig/}}
\setstretch{0.96}

	\title{Distributed Generalized Nash Equilibrium Seeking for Energy Sharing Games}
	
	\author{Zhaojian~Wang, 
		Feng~Liu,
		Zhiyuan~Ma,
		Yue Chen,
		Mengshuo Jia,
		Wei Wei,
		and Qiuwei Wu
}


	\maketitle
	
	\begin{abstract}                                               	
		With the proliferation of distributed generators and energy storage systems, traditional passive consumers in power systems have been gradually evolving into the so-called ``prosumers", i.e., proactive consumers, which can both produce and consume power. To encourage energy exchange among prosumers, energy sharing is increasingly adopted, which is usually formulated as a generalized Nash game (GNG). In this paper, a distributed approach is proposed to seek the Generalized Nash equilibrium (GNE) of the energy sharing game. To this end, we convert the GNG into an equivalent optimization problem. A Krasnosel'ski{\v{i}}-Mann iteration type algorithm is thereby devised to solve the problem and consequently find the GNE in a distributed manner. The convergence of the proposed algorithm is proved rigorously based on the nonexpansive operator theory. The performance of the algorithm is validated by experiments with three prosumers, and the scalability is tested by simulations using 123 prosumers.
	\end{abstract}

\section{Introduction}
It has been recognized that our power system is undergoing a fundamental transition
due to: 1) The proliferation of distributed generations (DGs), such as wind turbines, photovoltaics, energy storage systems and electric vehicles \cite{Chiang2010Building,Mani2019digital,Wang2017Optimal,Tran2012Realizing,Zhang2019Vehicle}; 2) The advancement of communications and control in consumer-level via smart appliances and energy management systems \cite{Han2010Smart,weng2018distributed,Zhang2019event}. Together, these changes allow traditional passive consumers to convert into the so-called  ``prosumers", i.e., proactive consumers, that can actively regulate their generation and consumption \cite{Dimeas2014Smart,Parag2016Electricity,Morstyn2018Using,Liu2017An,Chen2019An}.
Conventionally, a hierarchical structure is utilized in the power system energy management, usually in a centralized way. However, the centralized manner may face great challenges raised by the ever-increasing number of prosumers and uncertainties of renewable generations. This essentially advocates a distributed paradigm \cite{Distributed_II:Wang,dorfler2016breaking,wang2017unified,wang2018distributed} in energy management. Particularly, energy sharing turns to be a promising form of the market that encourages energy trading among prosumers.  As prosumers typically belong to different owners, the inherent competition may lead to strategic behaviors in energy sharing. In this situation, each prosumer intends to optimize its own profit while maintaining the power balance over the whole system. This leads to a generalized Nash game (GNG) with global constraints. In this regard, it is desirable to investigate distributed approaches to seeking the generalized Nash equilibrium (GNE) of the GNG. 

There are many papers investigating distributed methods of GNE seeking, which can be roughly divided into two categories in terms of methods used: gradient-based algorithms \cite{liang2017distributed,ZENG201920,yi2017distributed,yi2019asynchronous,Yu2017Distributed_learning,Paccagnan2019Nash}, and proximal-point algorithms \cite{SALEHISADAGHIANI20176166,yi2018distributed_layered,Belgioioso2019a,SALEHISADAGHIANI2019Distributed}. 
In the first type, the pseudo-gradient of each player's disutility function is utilized to seek the GNE, including continuous-time algorithms  \cite{liang2017distributed,ZENG201920} and discrete algorithms \cite{yi2017distributed,yi2019asynchronous,Yu2017Distributed_learning,Paccagnan2019Nash}. In \cite{liang2017distributed}, a distributed continuous-time projection-based algorithm is proposed to seek the GNE of aggregative games with linear coupled constraints. This is extended to the case where players have nonsmooth payoff functions in \cite{ZENG201920}. In \cite{yi2017distributed}, two distributed primal-dual algorithms are proposed for computing a GNE in noncooperative games with shared affine constraints. The operator splitting method is utilized to prove the convergence of the algorithms. It is further improved in \cite{yi2019asynchronous} to consider communication time delays and partial-decision information. In \cite{Yu2017Distributed_learning}, three stochastic gradient strategies are developed to seek GNE where agents are subject to randomness in the environment of unknown statistical distribution. In \cite{Paccagnan2019Nash}, the relations between Nash and Wardrop equilibria of the aggregative game are investigated and two algorithms are proposed to seek the equilibrium. 

In the proximal-point algorithms, the Alternating Direction Method of Multipliers (ADMM) method is widely used. An inexact-ADMM algorithm is proposed in \cite{SALEHISADAGHIANI20176166}. This method is improved in \cite{SALEHISADAGHIANI2019Distributed}, where each player only has partial information of their opponents and the communication graph is not necessarily the same as the cost dependency graph of each player. In \cite{yi2018distributed_layered}, two double-layer preconditioned proximal-point algorithms are proposed to seek GNE with both coupled equality and inequality constraints, respectively. For the GNG with a special structure, i.e., the coupling in the cost functions of the agents is linear, a distributed proximal-point algorithm is developed in \cite{Belgioioso2019a} for GNG with maximally monotone pseudo-gradient.

The aforementioned works have made great progress in the context of distributed GNE seeking. In most of the works, if the objective function of one player is associated with decisions of all players, each player is required to communicate with all of the other players, i.e., full information is needed. Then, the communication network is very dense. To address this problem, a local estimation of the overall decision profile for each player is added in \cite{pavel2019distributed,yi2019asynchronous}. However, this will increase one order in the algorithm, making it more complicated. 

In this paper, we offer a different perspective for seeking the GNE in the GNG when full information is included in the objective functions. The energy sharing game is transformed into an equivalent optimization problem. Instead of dealing with the GNG directly, we alternatively solve an equivalent optimization problem of the original game problem. Both experiments and simulations are utilized to validate our method. The main contributions are as follows.
\begin{itemize}
	\item A systematic way is proposed to transform the energy sharing game into an equivalent centralized optimization problem, which could be solved in a fully distributed way with neighboring communication. Moreover, the existence and uniqueness of the GNE of the energy sharing game are proved. 
	\item A distributed method is devised based on Krasnosel'ski{\v{i}}-Mann iteration to solve the equivalent counterpart instead of solving the original game problem directly. By constructing a firmly nonexpansive operator, we prove that the proposed distributed algorithm converges to the GNE of the original energy sharing game. 
\end{itemize}


The rest of this paper is organized as follows. In Section \ref{Preliminaries}, we briefly introduce some necessary preliminaries. 
Section \ref{model} formulates the energy sharing model. In Section \ref{Properties}, the existence and uniqueness of GNE are analyzed. A distributed GNE seeking algorithm is proposed in Section \ref{Algorithm}. The convergence of the algorithm is proved in Section \ref{Convergence}. The effectiveness of the algorithm is verified in Section \ref{Experiments} by experiments and simulation studies. Section \ref{Conclusion} concludes the paper.

\section{Preliminaries}\label{Preliminaries}
In this paper, use $\mathbb{R}^n$ to denote the $n$-dimensional Euclidean space.  For a matrix $A$, $[A_{ij}]$ is the entry in the $i$-th row and $j$-th column of $A$.
For vectors $x,y\in \mathbb{R}^n$, $x^{\mathrm{T}}y=\left\langle x,y \right\rangle$ denotes the inner product of $x,y$. 
$\left\|x \right\|_2=\sqrt{x^{\mathrm{T}}x}$ denotes the Euclidean norm of $x$. 
Denote the inner product under a positive definite matrix $Q$ by $\left\langle x,y \right\rangle_Q=\left\langle Q x,y \right\rangle$. Similarly, the norm induced by $Q$ is $\left\|x \right\|_Q=\sqrt{\left\langle Q x,x \right\rangle}$. 
The following relationship holds for a $Q$-induced norm.
\begin{equation}\label{relationship}
\|a-c\|_{Q}^{2}-\|b-c\|_{Q}^{2}=2\langle a-b, a-c\rangle_{Q}-\|a-b\|_{Q}^{2}
\end{equation}
which can be obtained by the equation $\|a+b\|_{Q}^{2}=$
$\|a\|_{Q}^{2}+2\langle a, b\rangle_{Q}+\|b\|_{Q}^{2}$.

The identity matrix with dimension $n$ is denoted by $I_n$. Use $\prod_{i=1}^n\Omega_i$ to denote the Cartesian product of the sets $\Omega_i, i=1, \cdots, n$. Define the projection of $x$ onto a set $\Omega$ as 
\begin{equation}
\label{def:projection}
\mathcal{P}_{\Omega}(x)=\arg \min_{y\in \Omega}\left\|x-y \right\|_2
\end{equation}
Use ${\rm{Id}}$ to denote the identity operator, i.e., ${\rm{Id}}(x)=x$, $\forall x$. Define $N_\Omega(x)=\{v|\left\langle v, y-x\right\rangle\le 0, \forall y\in \Omega\}$. We have $\mathcal{P}_{\Omega}(x)=({\rm{Id}}+N_{\Omega})^{-1}(x)$  \cite[Chapter 23.1]{bauschke2011convex}.

%


For a single-valued operator $\mathcal{T}:\Omega\subset\mathbb{R}^n\rightarrow\mathbb{R}^n$, a point $x\in \Omega$ is a fixed point of $\mathcal{T}$ if $\mathcal{T}(x)\equiv x$. The set of fixed points of $\mathcal{T}$ is denoted by ${Fix}(\mathcal{T})$. $\mathcal{T}$ is nonexpansive if $\left\|\mathcal{T}(x)- \mathcal{T}(y)\right\|\le\left\|x- y\right\|, \forall x, y \in \Omega$. 
For $\alpha\in (0,1)$, $\mathcal{T}$ is called $\alpha$-averaged if there exists a nonexpansive operator $\mathcal{R}$ such that $\mathcal{T}=(1-\alpha){\rm{Id}}+\alpha \mathcal{R}$. Use $\mathcal{A}(\alpha)$ to denote the class of $\alpha$-averaged operators. If $\mathcal{T}\in \mathcal{A}(\frac{1}{2})$, $ \mathcal{T} $ is called firmly nonexpansive. The graph of $\mathcal{T}$ is $\operatorname{gra} \mathcal{T}=\left\{(x, u) \in \mathbf{R}^{n} \times \mathbf{R}^{n} | u \in \mathcal{T} (x)\right\}$. $\mathcal{T}$ is monotone if $\forall(x, u),\ \forall(y, v) \in \operatorname{gra} \mathcal{T},\langle x-y, u-v\rangle \geq0$. $ \mathcal{T} $ is maximally monotone if $\operatorname{gra} \mathcal{T}$ is not strictly contained in the graph of any other monotone operator.

\section{Energy Sharing Game}\label{model}
In this section, the model of prosumers is introduced. Then, we formulate the energy sharing game. 
\subsection{Model of prosumers}
In this paper, we consider the energy sharing problem among a set of prosumers, denoted by $ \mathcal{N}=\{1,2,...n\} $. The communication edge is denoted by $\mathcal{E}\subseteq \mathcal{N}\times \mathcal{N}$. For a prosumer $ i $, the set of its neighbors is denoted by $ N_i $. If $ j\in N_i $, prosumers $ i $ and $ j $ can communicate directly. The Laplacian matrix of the communication graph is denoted by $L$ and we have $\textbf{1}^{\rm T}L=0$, where $ \textbf{1} $ is a vector with all of the components as $1$.

The scenario is that each prosumer has a power shortage or surplus, denoted by $D_{i}$. To balance its power, it can produce power itself, denoted by $ p_i $, or buy power from (sell power to) the grid, denoted by $ q_i $. 
In the sharing market, the demand function of each prosumer can be expressed by
\begin{equation}
q_i=a_i\mu_c+b_i
\end{equation}
where $\mu_c$ is the market clearing price, $ b_i $ is the willingness to buy. $ a_i < 0 $ implies the price elasticity.
The market clears when the net quantity $ \sum\nolimits_{i} q_i =0 $ and the obtained sharing price is\footnote{Given a collection of $z_i$ for $i$ in a certain set $Z$, $z$ denotes the column vector
	$z := (z_i, i\in Z)$ with a proper dimension with $z_i$ as its components.}
\begin{equation}\label{price}
\mu_c=-\frac{\textbf{1}^{\rm T}b}{\textbf{1}^{\rm T}a}=-\frac{\textbf{1}^{\rm T}b}{N\overline a}
\end{equation}
where $ \overline a=\frac{\textbf{1}^{\rm T}a}{N} $ is the average value of $a_i$. In the rest of the paper, $b_{-i}:=(b_1, b_2, \cdots, b_{i-1}, b_{i+1}, \cdots, b_n)^{\rm T}$.

\subsection{Energy sharing game}
 
In the energy sharing game, each prosumer intends to minimize its cost while maintaining the global power balance. The optimization problem of each prosumer is
\begin{subequations}
	\label{Sharing_problem}     
	\begin{align}
	\min\limits_{p_i,b_i}\quad& f_i(p_i,b)= h_i(p_i)+(a_i\mu_c(b)+b_i)\mu_c(b)
	\label{Sharing_problem1}
	\\ 
	\text{s.t.} \quad
	& p_i+a_i\mu_c(b)+b_i=D_i 
	\label{Sharing_problem2}
	\\
	\label{Sharing_problem3}
	&\sum\nolimits_{i\in\mathcal{N}} (a_i\mu_c(b)+b_i) =0 
	\\
	\label{Sharing_problem4}
	&\underline p_i\le p_i\le \overline p_i  
	\end{align}
\end{subequations}
The disutility function consists of two parts, where $ h_i(p_i) $ is the cost of prosumer $ i $ to produce $p_i$ and $ (a_i\mu_c(b)+b_i)\mu_c(b)= q_i\mu_c(b)$ is the cost of buying power $ q_i $. 

{\color{black}
	The Lagrangian of \eqref{Sharing_problem} is 
	\begin{align}\label{Lagrangian}
	&\mathop{\mathcal {L}(p_i,b,\lambda_i,\eta_i)}\limits_{\underline p_i\le p_i\le \overline p_i,b_i}=f_i(p_i,b)  + \lambda_i\left(p_i+a_i\mu_c(b)+b_i-D_i\right)\nonumber\\
	&\quad+ \eta_i\left(\sum\nolimits_{i\in\mathcal{N}} (a_i\mu_c(b)+b_i)\right)
	\end{align}
	where $\lambda_i, \eta_i$ are the Lagrangian multipliers.
}

Regarding the energy sharing game \eqref{Sharing_problem}, we have following assumptions.
\begin{assumption}\label{convex}
	The functions $ h_i(p_i) $ are convex and differentiable.
\end{assumption}
\begin{assumption}\label{Slater}
	For a given $ b_{-i} $, the Slater's condition of the problem \eqref{Sharing_problem} holds \cite[Chapter 5.2.3]{boyd2004convex}, i.e., problem \eqref{Sharing_problem} is feasible.
\end{assumption}
\begin{assumption}\label{uniqueness}
	The generation satisfies $\textbf{1}^{\rm T}\underline p< \textbf{1}^{\rm T}D<\textbf{1}^{\rm T} \overline p $. 
\end{assumption}
For Assumption 1, the power generation cost $ h_i(p_i) $ is usually quadratic, which is convex and differentiable. 
For Assumption 2, it is a common assumption, otherwise the problem is infeasible.  
The Assumption 3 implies that the energy production can strictly satisfy the load demand. As long as the total regulation capacity of all prosumers is strictly larger than the total load demand, Assumption 3 holds. Assumption 3 is stricter than Assumption 2, which is used to determine the uniqueness of the GNE.

In summary, the energy sharing game among all prosumers is composed of the following elements:
\begin{itemize}
	\item Player: all prosumers, denoted by $ \mathcal{N}=\{1,2,...n\} $;
	\item Strategy: generation power $p_i$ and power interaction $ b_i $;
	\item Payoff: the disutility function $ f_i(p_i,b), \forall i\in\mathcal{N} $.
\end{itemize}
The GNE $ (p^*,b^*) $ of the game in \eqref{Sharing_problem} is defined as:
\begin{align*}
(p_i^*,b_i^*)\in \arg\min f_i(p_i,b_i,b_{-i}^*),\ s.t.\ \eqref{Sharing_problem2},\eqref{Sharing_problem2},\eqref{Sharing_problem3},\  \forall i\in\mathcal{N}
\end{align*}

\begin{remark}
	The energy sharing game model \eqref{Sharing_problem} basically comes from \cite{Chen2019An}, and is modified by additionally considering local capacity constraints \eqref{Sharing_problem4}. Mathematically, the energy sharing game \eqref{Sharing_problem} is indeed a GNG by noting that the global constraint  \eqref{Sharing_problem3} couples the strategies of players. As proved in \cite{Chen2019An}, if the constraint \eqref{Sharing_problem4} is not taken into account, the energy sharing game can be equivalently converted into a standard Nash game instead of a GNG without any coupling constraints. However,  when \eqref{Sharing_problem4} is considered, the coupling constraints cannot be eliminated. This is the key difference between this paper and \cite{Chen2019An} in the modeling. 
\end{remark}

\section{Equivalent Problem and Uniqueness of GNE}\label{Properties}
In this section, we analyze the property of the Nash equilibrium and transform the game into an equivalent optimization problem.

Replacing $\mu_c$ in \eqref{Sharing_problem} using \eqref{price}, then  the sharing game problem \eqref{Sharing_problem} can be rewritten as 
\begin{subequations}
	\label{Sharing_problem_II}     
	\begin{align}
	\min\limits_{p_i,b_i}\quad& f_i(p_i,b)= h_i(p_i)+g_i(b_i,b_{-i})
	\label{Sharing_problem1_II}
	\\ 
	\text{s.t.} \quad
	&  p_i+\frac{\textbf{1}^{\rm T}a-a_i}{\textbf{1}^{\rm T}a}b_i-\frac{a_i}{\textbf{1}^{\rm T}a}\sum\nolimits_{j\neq i}b_j=D_i
	\label{Sharing_problem2_II}
	\\
	\label{Sharing_problem4_II}
	&\underline p_i\le p_i\le \overline p_i  
	\end{align}
\end{subequations}
where 
\begin{align}
g_i(b_i,b_{-i})&= (a_i\mu_c(b)+b_i)\mu_c(b)\nonumber\\
&=-\left(\frac{\textbf{1}^{\rm T}a-a_i}{\textbf{1}^{\rm T}a}b_i-\frac{a_i}{\textbf{1}^{\rm T}a}\sum\limits_{j\neq i}b_j\right)\frac{b_i+\sum\nolimits_{j\neq i}b_j}{\textbf{1}^{\rm T}a}\nonumber\\
&=\frac{a_i-\textbf{1}^{\rm T}a}{N^2\overline a^2}b^2_i+\left(\frac{2a_i-\textbf{1}^{\rm T}a}{N^2\overline a^2}\right)b_i\sum\nolimits_{j\neq i}b_j \nonumber\\
&\quad+\frac{a_i}{N^2\overline a^2}\bigg(\sum\nolimits_{j\neq i}b_j\bigg)^2\nonumber\\
&=\alpha_ib^2_i+\beta_ib_i\sum\nolimits_{j\neq i}b_j +\delta_i\bigg(\sum\nolimits_{j\neq i}b_j\bigg)^2
\end{align} 
with 
\begin{align}\label{alpha_definition}
	\alpha_i=\frac{a_i-\textbf{1}^{\rm T}a}{N^2\overline a^2}, \beta_i=\frac{2a_i-\textbf{1}^{\rm T}a}{N^2\overline a^2},  \delta_i=\frac{a_i}{N^2\overline a^2}
\end{align}
Because $a_i<0$, we have  $ \alpha_i>0, \beta_i>0, \delta_i<0 $. Then, $ 	g_i(b_i,b_{-i}) $ is strictly convex with respect to $ b_i $ given any $ b_{-i} $. 
The objective function $ g_i(b_i,b_{-i}) $ is associated with all of other players. When solving \eqref{Sharing_problem_II} directly, it requires each player to be able to access all other players’ decisions, i.e., full information is needed. This is very difficult to implement in practice. In the rest of this work, we give a different perspective by solving an equivalent optimization problem of \eqref{Sharing_problem_II}. 

The pseudo-gradient of $ f_i(p_i,b), \forall i\in\mathcal{N} $, denoted by $ F(p,b)$, is defined as 
\begin{align}\label{pseudo_gradient}
F(p,b)=col\left(\nabla_{p_1,b_1}f_1(p_1,b),\cdots,\nabla_{p_N,b_N}f_N(p_N, b)\right) 
\end{align}
where $ col\left(x_1,\cdots,x_n\right) $ is the column vector stacked with column vectors $ x_1,\cdots,x_n $.

Define the sets 
\begin{equation}
\Omega_i:=\left\{p_i\ |\ \underline p_i\le p_i\le \overline p_i\right\},\ \Omega=\prod\nolimits_{i=1}^n\Omega_i
\end{equation}
and 
\begin{equation}
X^e_i:=\left\{(p_i, b_i)\ |\ p_i+\frac{\textbf{1}^{\rm T}a-a_i}{\textbf{1}^{\rm T}a}b_i-\frac{a_i}{\textbf{1}^{\rm T}a}\sum\nolimits_{j\neq i}b_j=D_i\right\} 
\end{equation}
\begin{equation}
X^e=\prod\nolimits_{i=1}^nX^e_i, \ X_i=\Omega_i\cap X_i^e, \ X=\Omega\cap X^e
\end{equation}

For any given $ b_{-i} $, the feasible set of prosumer $ i $, $ X_i$, is closed and convex.  Then, we have the following result.
\begin{lemma}\label{lemma_VI}
	If assumptions \ref{convex} and \ref{Slater} hold, a point $ (p, b) $ is an equilibrium if and only if it is a solution
	of the variational inequality $VI(X, {F}(p, b))$\footnote{The variational inequality problem  $VI(X, {F}(x))$ is to find a vector $\overline{x} \in X$ such that
		$(y-\overline{x})^{T} {F}(\overline{x}) \geq 0$ for all $y \in X $\cite{facchinei2010generalized}.}.
\end{lemma}
\begin{proof}
	For any given $ b_{-i} $, the disutility function \eqref{Sharing_problem1_II} is continuous, differentiable and convex with respect to $ (p_i, b_{i}) $. Moreover, the feasible set $ X_i $ is closed and convex. Thus, we have this assertion by \cite[Theorem 3.3]{facchinei2010generalized}. 
\end{proof}

For the existence and uniqueness of the GNE, we have the following results. 
\begin{theorem}
	\label{unique}
	If Assumption \ref{convex} and Assumption \ref{Slater} hold, for the generalized Nash game \eqref{Sharing_problem_II}, we have
	\begin{enumerate}
		\item the generalized Nash equilibrium exists;
		\item if the Assumption \ref{uniqueness} also holds, the generalized Nash equilibrium is unique.
	\end{enumerate}
\end{theorem}
The detailed proof of Theorem \ref{unique} is given in Appendix \ref{Appendix_Unique}. From the proof of Theorem \ref{unique}, an equivalent optimization problem is obtained, which is
\begin{subequations}
	\label{Central_problem_II}     
	\setlength{\abovedisplayskip}{4pt}	
	\setlength{\belowdisplayskip}{4pt}
	\begin{align}
	\min\limits_{p}\ & \hat h(p)=\sum\limits_{i\in\mathcal{N}} \left(h_i(p_i)+ \frac{p_i^2}{2(a_i-\textbf{1}^{\rm T}a)}-\frac{D_i}{a_i-\textbf{1}^{\rm T}a}p_i\right)
	\label{Central_problem1_II}
	\\ 
	\text{s.t.} \quad
	&  \sum\nolimits_{i\in\mathcal{N}} p_i =\sum\nolimits_{i\in\mathcal{N}} D_i,\quad \mu_c
	\label{Central_problem2_II}
	\\
	\label{Central_problem3_II}
	&\underline p_i\le p_i\le \overline p_i  
	\end{align}
\end{subequations}
where $\mu_c$ is the Lagrangian multiplier of constraint \eqref{Central_problem2_II}. The ``equivalent" here means that we can first solve \eqref{Central_problem_II} to get $p^*, \mu_c^*$, and then get $b^*$ with $ b_i^{*}=D_i-p^{*}_i-a_i\mu_c^* $. Finally, the GNE is obtained.

\begin{remark}\label{former_work}
	The game \eqref{Sharing_problem_II} is difficult to solve by existing methods due to the full information in the objective function. Now, we can alternatively consider the equivalent counterpart \eqref{Central_problem_II}, which could be solved in a distributed way as long as $ \textbf{1}^{\rm T}a $ is known. In this regard, the market coordinator (or a third-part platform) is required to broadcast the sum of $ a_i $ of all prosumers. That is needed merely when new prosumers join or existing ones quit, which can be known by the market coordinator or the third-part platform. After getting $p_i^*$, $ b_i^* $ can be obtained from $ b_i^{*}=D_i-p^{*}_i-a_i\mu_c^* $. Then, the energy sharing game problem \eqref{Sharing_problem_II} can be solved. 
\end{remark}


\section{Distributed Algorithm for Equilibrium Seeking}\label{Algorithm}

In this section, we first propose a  distributed algorithm based on Krasnosel'ski{\v{i}}-Mann iteration to solve the problem \eqref{Central_problem_II}, i.e., to solve the generalized game \eqref{Sharing_problem_II}. Then, we use the nonexpansive operator theory to prove the convergence of the proposed algorithm.

\subsection{Algorithm design}
Before giving the algorithm, we first define a matrix and a function. 
Define the matrix
\begin{equation}\label{Theta}
\Theta:=\left[ {\begin{array}{*{20}{c}}
	\Gamma &0 &-I_n\\
	0& \alpha_{z}^{-1}I_{n}& -L\\
	-I_n & -L &\alpha_\mu^{-1}I_n
	\end{array}} \right]
\end{equation}
where the matrix $ \Gamma=\diag(\gamma_i) $. 
$\gamma_i$, $\alpha_{z}$ and $\alpha_\mu$ are constant to make $\Theta$ positive definite\footnote{It should be note that such $ \gamma_i, \sigma_{z}, \sigma_{\mu} $ always exist to make $ \Theta $ diagonally dominant, or positive definite. }. 

Define the function 
\begin{equation}
H_i(p_i)=\frac{\partial \tilde{h}_i}{\partial p_i}(p_i)+\gamma_ip_i
\end{equation}
where $ \tilde{h}_i(p_i)= h_i(p_i)+ \frac{p_i^2}{2(a_i-\textbf{1}^{\rm T}a)}-\frac{D_i}{a_i-\textbf{1}^{\rm T}a}p_i $. We have the follow result.
\begin{lemma}
	The function $ H_i(p_i) $ is strictly monotone. Moreover, its inverse function $ H_i^{-1}(p_i) $ exists and is also strictly monotone.
\end{lemma}
The proof is straightforward as $ \tilde{h}_i(p_i) $ is strongly convex, which is omitted here. 

We propose the following algorithm, which is denoted by SGNE (\underline{S}eeking the \underline{G}eneralized \underline{N}ash \underline{E}quilibrium).

\begin{algorithm}[t]
	\caption{\textit{SGNE}}
	\label{algorithm0}
	%
	\begin{subequations}\label{inertia_algorithm} 
		\begin{flushleft}
			\textbf{Prediction phase}:
		\end{flushleft}
		For prosumer $ i $, it computes
		\begin{align}
		\tilde{p}_{i,t}&={p}_{i,t}+\eta({p}_{i,t}-{p}_{i,t-1}) \\		
		\tilde{z}_{i,t}&={z}_{i,t}+\eta({z}_{i,t}-{z}_{i,t-1})\\		
		\tilde{\mu}_{i,t}&={\mu}_{i,t}+\eta({\mu}_{i,t}-{\mu}_{i,t-1})
		\end{align}	
		\begin{flushleft}
			\textbf{Update phase}:
		\end{flushleft}
		For prosumer $ i $, it computes  
		\begin{equation}
		\label{algorithm5}
		p_{i,t+1}= \left[H^{-1}(\gamma_i\tilde p_{i,t}-\tilde\mu_{i,t})\right]_{\underline p_i}^{\overline p_i}
		\end{equation}		
		Communicate with its neighbors $ j\in N_i $ to get $ \tilde\mu_{j,t} $, and compute
		\begin{equation}
		\label{algorithm6}						
		z_{i,t+1}=\tilde z_{i,t}-\sigma_{z} \sum\nolimits_{j\in {N}_i} (\tilde\mu_{i,t}-\tilde\mu_{j,t})
		\end{equation}
		Communicate with its neighbors $ j\in N_i $ to get $ \tilde z_{j,t}, z_{j,t+1} $, and compute
		\begin{align}
		\label{algorithm7}
		&\mu_{i,t+1}=\tilde \mu_{i,t}+\sigma_{\mu}\bigg( 2p_{i,t+1}-\tilde p_{i,t} - D_i  \nonumber\\
		&\quad +2\sum\limits_{j\in {N}_i} (z_{i,t+1}-z_{j,t+1})-  \sum\limits_{j\in {N}_i} (\tilde z_{i,t}-\tilde z_{j,t})  \bigg) 
		\end{align}
	\end{subequations}
\end{algorithm}
The algorithm has the Krasnosel'ski{\v{i}}-Mann iteration type \cite[Chapter 5]{bauschke2011convex}, which consists of two phases: prediction phase and update phase. In the prediction phase, each bus uses the local stored information of the last two steps to get predictive variables by a linear extrapolation. In the update phase, the predictive variables are utilized to proceed the next iteration. In the algorithm, only communications with neighbors are needed, which means that it is fully distributed.

\subsection{Algorithm reformulation}
From \eqref{algorithm5}, we know 
\begin{equation}
H^{-1}(\gamma_i\tilde p_{i,t}-\tilde\mu_{i,t})  \left\{ \begin{array}{l}
\le p_{i,t+1}, \ p_{i,t+1}= \underline p_i\\
=p_{i,t+1},\ \underline p_i< p_{i,t+1}< \overline p_i  \\
\ge p_{i,t+1},\ p_{i,t+1}= \overline p_i
\end{array} \right.
\end{equation}
Because $ H_i^{-1}(p_i) $ is monotone, we have 
\begin{equation}
H(p_{i,t+1})-\gamma_i\tilde p_{i,t}+\tilde\mu_{i,t}  \left\{ \begin{array}{l}
\ge 0, \ p_{i,t+1}= \underline p_i\\
=0,\ \underline p_i< p_{i,t+1}< \overline p_i  \\
\le 0,\ p_{i,t+1}= \overline p_i
\end{array} \right.
\end{equation}
It is equivalent to
\begin{align}
p_{i,t+1}&=\mathcal{P}_{\Omega_i}\left(p_{i,t+1}-\alpha_{p}(\nabla\tilde h_i(p_{i,t+1})\right. \nonumber\\
&\left.\qquad\qquad\qquad+\gamma_i(p_{i,t+1}-\tilde p_{i,t})+\tilde\mu_{i,t}\right)
\end{align}
Recalling $\mathcal{P}_{\Omega}(x)=({\rm{Id}}+N_{\Omega})^{-1}(x)$, we have 
\begin{align}
\gamma_i(\tilde p_{i,t}-p_{i,t+1})&-\tilde\mu_{i,t}+\mu_{i,t+1}\in N_{\Omega_i}(p_{i,t+1}) \nonumber\\
&+\nabla\tilde h_i(p_{i,t+1})+ \mu_{i,t+1}
\end{align}
From \eqref{algorithm6}, we have 
\begin{equation}
\begin{split}
&\sigma^{-1}_{z}(\tilde z_{i,t}-z_{i,t+1}) -\sum\nolimits_{j\in {N}_i} (\tilde\mu_{i,t}-\tilde\mu_{j,t}) \\
&\quad+\sum\limits_{j\in {N}_i} (\mu_{i,t+1}-\mu_{j,t+1}) = \sum\limits_{j\in {N}_i} (\mu_{i,t+1}-\mu_{j,t+1})
\end{split}
\end{equation}
From \eqref{algorithm7}, we have 
\begin{align}
&\sigma_{\mu}^{-1}(\tilde \mu_{i,t}-\mu_{i,t+1})+ \sum\nolimits_{j\in {N}_i} ( z_{i,t+1}- z_{j,t+1}) \nonumber \\ 
&\quad - \sum\nolimits_{j\in {N}_i} (\tilde z_{i,t}-\tilde z_{j,t}) + p_{i,t+1}-\tilde p_{i,t} \nonumber \\
&\quad   =- p_{i,t+1} + D_i-\sum\nolimits_{j\in {N}_i} ( z_{i,t+1}- z_{j,t+1})
\end{align}

The compact form of the algorithm is 
\begin{subequations}\label{compact_algorithm}
	\begin{align}
	&\tilde{p}_t={p}_t+\eta({p}_t-{p}_{t-1})\\
	&\tilde{z}_t={z}_t+\eta({z}_t-{z}_{t-1})\\
	&\tilde{\mu}_t={\mu}_t+\eta({\mu}_t-{\mu}_{t-1})\\
	&\Gamma(\tilde p_{t}-p_{t+1})-(\tilde\mu_{t}-\mu_{t+1})\in N_{\Omega}(p_{t+1})\nonumber\\
	&\qquad\qquad\qquad\qquad\qquad+\nabla\tilde h(p_{t+1})+ \mu_{t+1}\\
	&\sigma^{-1}_{z}(\tilde z_{t}-z_{t+1})-L(\tilde\mu_t-\mu_{t+1})=L\mu_{t+1}\\
	&\sigma^{-1}_{\mu}(\tilde \mu_{t}-\mu_{t+1})-L(\tilde z_{t}-z_{t+1})\nonumber\\
	&\qquad\qquad-(\tilde p_t-p_{t+1})=-p_{t+1}+D-Lz_{t+1}
	\end{align}
\end{subequations}

Define the following operator
\begin{equation}
\label{OperU}
\mathcal{U}:\left[ {\begin{array}{*{20}{c}}
	{p}\\
	{{z}}\\
	{{\mu}}
	\end{array}} \right] \mapsto \left[ {\begin{array}{*{20}{c}}
	\mu + N_{\Omega}(p) +\nabla\tilde h(p) \\
	L\mu\\
	-p+D-Lz
	\end{array}} \right]
\end{equation}
Then, the algorithm \eqref{compact_algorithm} is rewritten as
\begin{subequations}\label{compact2}
	\begin{align}\label{compact_e1}
	&\tilde\omega_t=\omega_t+\eta(\omega_t-\omega_{t-1})\\
	\label{compact_e2}
	&\Theta(\tilde\omega_t-\omega_{t+1})\in\mathcal{U} (\omega_{t+1})
	\end{align}
\end{subequations}
where $\Theta$ is defined in \eqref{Theta}.
For the operator $\mathcal{U}$, it has the following properties
\begin{lemma}\label{Lemma1}
	Take step sizes  $\Gamma$, $\alpha_{z}$ and $\alpha_\mu$ such that $\Theta$ is positive definite We have following properties. 
	\begin{enumerate} 
		\item Operator $\mathcal{U}$ is maximally monotone;
		\item $\Theta^{-1}\mathcal{U}$ is maximally monotone under the $\Theta$-induced norm $\|\cdot\|_{\Theta}$;
		\item $({\rm{Id}}+\Theta^{-1}\mathcal{U})^{-1}$ exists and is firmly nonexpansive.
	\end{enumerate}
\end{lemma}
\begin{proof}
	1): The operator $\mathcal{U}$ can be rewritten as 
	\begin{eqnarray}
	\mathcal{U}&=&\left[ {\begin{array}{*{20}{c}}
		0 &0 &I\\
		0& 0& L\\
		-I & -L &0
		\end{array}} \right]\left[ {\begin{array}{*{20}{c}}
		{p}\\
		{{z}}\\
		{{\mu}}
		\end{array}} \right]+\left[ {\begin{array}{*{20}{c}}
		N_{\Omega}(p) +\nabla\tilde h(p)\\
		0\\
		D
		\end{array}} \right]\nonumber\\
	&=&\mathcal{U}_1+\mathcal{U}_2
	\end{eqnarray}
	As $\mathcal{U}_1$ is a skew-symmetric matrix, $\mathcal{U}_1$ is maximally monotone \cite[Example 20.35]{bauschke2011convex}. Moreover, $N_{\Omega}(\textbf{z})$ is maximally monotone by \cite[Example 20.26]{bauschke2011convex}. In addition, $0$ and $ D $ are monotone and continuous, which are all maximally monotone \cite[Corollary 20.28]{bauschke2011convex}. As $ \tilde h(p) $ is convex, $ \nabla\tilde h(p) $ is also maximally monotone \cite[Theorem 20.25]{bauschke2011convex}. Thus, $\mathcal{U}_2$ is also maximally monotone. Then, by \cite[Corollary 25.5]{bauschke2011convex},we have $\mathcal{U}$ is maximally monotone.
	
	2) As $\Theta$ is symmetric positive definite and $\mathcal{U}$ is maximally monotone, we can prove that $\Theta^{-1}\mathcal{U}$ is maximally monotone under the $\Theta$-induced norm by the similar analysis in Lemma 5.6 of \cite{yi2017distributed}.  
	
	3) As $\Phi^{-1}\mathcal{U}$ is maximally monotone, $({\rm{Id}}+\Phi^{-1}\mathcal{U})^{-1}$ exists and is firmly nonexpansive by \cite[Proposition 23.8]{bauschke2011convex}. 	
\end{proof}

By the third assertion of Lemma \ref{Lemma1}, \eqref{compact2} is equivalent to 
\begin{subequations}\label{compact_f}
	\begin{align}\label{compact_f1}
	&\tilde\omega_t=\omega_t+\eta(\omega_t-\omega_{t-1})\\
	\label{compact_f2}
	&\omega_{t+1}=(\rm Id+\Theta^{-1}\mathcal{U})^{-1}\tilde \omega
	\end{align}
\end{subequations}
Up to now, we transform the algorithm SGNE to a fixed point problem with a nonexpansive operator $ (\rm Id+\Theta^{-1}\mathcal{U})^{-1} $, which provides fundamental support for the convergence proof.

\section{Nash Equilibrium Seeking and Convergence}\label{Convergence}
In this section, we address the optimality of the equilibrium point and the convergence of the algorithm SGNE, i.e., the discrete dynamic system \eqref{inertia_algorithm}.
\subsection{Nash Equilibrium}
First, we define the equilibrium of the algorithm SGNE.
\begin{definition}
	\label{def:ep.1}
	A point $w^*=(w^*_i, i\in\mathcal{N})=(p^*_i, z^*_i, \mu^{*}_i)$ 
	is an {equilibrium point} of \eqref{inertia_algorithm} if $\lim\limits_{t\rightarrow +\infty} w_{i,t}=w_i^*$ holds for all $i$.
\end{definition}
Recalling the \eqref{EQb}, and we have the following result.
\begin{theorem}\label{Equivalence}
	Suppose assumptions \ref{convex}, \ref{Slater} and \ref{uniqueness} hold. At the equilibrium of the algorithm SGNE, we have $ \mu^{*}_i=\mu^{*}_j=\mu_0, \forall i,j $ is the clearing price and $ \left(p^*, b^* \right)   $ is the GNE of the game, where $ \mu_0 $ is constant.
\end{theorem}
\begin{proof}
	From Definition \ref{def:ep.1} and \eqref{compact_algorithm}, we have 
	\begin{subequations}\label{equilibrium}
		\begin{align}\label{equilibrium1}
		&0\in N_{\Omega}(p^*) +\nabla\tilde h(p^*)+ \mu^*\\
		\label{equilibrium2}
		&0=L\mu^*\\
		\label{equilibrium3}
		&0=-p^*+D-Lz^*
		\end{align}
	\end{subequations}
	From \eqref{equilibrium2}, we have 
	\begin{equation}\label{equilibrium4}
	\mu^{*}_i=\mu^{*}_j=\mu_0, \forall i,j  
	\end{equation}
	From \eqref{equilibrium3}, we have 
	\begin{equation}\label{equilibrium5}
	0=-\textbf{1}^{\rm T}p^*+\textbf{1}^{\rm T}D-\textbf{1}^{\rm T}Lz^*=-\textbf{1}^{\rm T}p^*+\textbf{1}^{\rm T}D
	\end{equation}
	The equations \eqref{equilibrium1}, \eqref{equilibrium4} and \eqref{equilibrium5} are the KKT condition \eqref{KKT_game_eq}. This completes the proof.
\end{proof}

\subsection{Convergence}
In this subsection, we analyze the convergence of the algorithm SGNE based on the compact form \eqref{compact_f}. First, we give the following result. 
\begin{theorem}
	Suppose Assumption \ref{convex} and Assumption \ref{Slater} hold. Given a parameter $\eta$ satisfying $0<\eta<\frac{1}{3}$ and the step sizes  $\Gamma$, $\alpha_{z}$ and $\alpha_\mu$ such that $\Theta$ is positive definite. Then with SGNE, $w_t$ converges to a primal-dual optimal solution $ \omega^* $ of the problem \eqref{Central_problem_II}. Then, $ (p^*, b^*) $ is the GNE of \eqref{Sharing_problem_II}.
\end{theorem}

\begin{proof}
	Frist, we prove that $ \lim\limits_{t\to\infty}(\omega_{t+1}-\omega_{t})=0 $. Given any equilibrium point $ \omega^* $, use the equation \eqref{relationship}, and we have 
	\begin{align}
	&\left\|\omega_t-\omega^*\right\|_\Theta^2-\left\|\omega_{t+1}-\omega^*\right\|_\Theta^2\nonumber\\
	&\qquad=2\langle\omega_t-\omega_{t+1}, \omega_t-\omega^*\rangle_\Theta- \left\|\omega_t-\omega_{t+1}\right\|_\Theta^2  \nonumber\\
	&\qquad=2\langle\omega_t-\omega_{t+1}, \omega_{t+1}-\omega^*\rangle_\Theta+ \left\|\omega_t-\omega_{t+1}\right\|_\Theta^2  \nonumber\\
	&\qquad=\left\|\omega_t-\omega_{t+1}\right\|_\Theta^2+2\langle\tilde\omega_t-\omega_{t+1}, \omega_{t+1}-\omega^*\rangle_\Theta\nonumber\\
	&\qquad\qquad\qquad-2\eta\langle\omega_t-\omega_{t-1},\omega_{t+1}-\omega^*\rangle_\Theta
	\end{align}
	where the last equation is derived by integrating \eqref{compact_f1}.
	
	By \eqref{compact_e2}, we have 
	\begin{equation}
	0\in\mathcal{U}(\omega^*)
	\end{equation}
	
	Since $ \mathcal{U} $ is maximally monotone, we have
	\begin{equation}
	\langle\Theta(\tilde\omega_t-\omega_{t+1}),\omega_{t+1}-\omega^*\rangle\ge 0
	\end{equation}
	Thus, we have 
	\begin{equation}	
	\begin{split}\label{middle1}
	&\left\|\omega_{t+1}-\omega^*\right\|_\Theta^2-\left\|\omega_t-\omega^*\right\|_\Theta^2\\
	&\qquad\le2\eta\langle\omega_t-\omega_{t-1},\omega_{t+1}-\omega^*\rangle_\Theta-\left\|\omega_t-\omega_{t+1}\right\|_\Theta^2
	\end{split}
	\end{equation}
	
	Moreover, 
	\begin{equation}	
	\begin{split}\label{middle2}
	&\eta(\left\|\omega_{t}-\omega^*\right\|_\Theta^2-\left\|\omega_{t-1}-\omega^*\right\|_\Theta^2)\\
	&\qquad=\eta(2\langle\omega_t-\omega_{t-1}, \omega_t-\omega^*\rangle_\Theta- \left\|\omega_t-\omega_{t-1}\right\|_\Theta^2 )
	\end{split}
	\end{equation}
	
	From \eqref{middle1} and \eqref{middle2}, we have
	\begin{align}\label{recursion}
	&\left\|\omega_{t+1}-\omega^*\right\|_\Theta^2-\left\|\omega_t-\omega^*\right\|_\Theta^2\nonumber\\
	&\qquad\qquad\qquad\qquad-\eta(\left\|\omega_{t}-\omega^*\right\|_\Theta^2-\left\|\omega_{t-1}-\omega^*\right\|_\Theta^2)\nonumber\\
	&\le 2\eta\langle\omega_t-\omega_{t-1},\omega_{t+1}-\omega^*\rangle_\Theta-\left\|\omega_t-\omega_{t+1}\right\|_\Theta^2\nonumber\\
	&\qquad-\eta(2\langle\omega_t-\omega_{t-1}, \omega_t-\omega^*\rangle_\Theta- \left\|\omega_t-\omega_{t-1}\right\|_\Theta^2 )\nonumber\\
	&=-\left\|\omega_t-\omega_{t+1}\right\|_\Theta^2+2\eta\langle\omega_t-\omega_{t-1},\omega_{t+1}-\omega_{t}\rangle_\Theta\nonumber\\
	&\qquad+\eta\left\|\omega_t-\omega_{t-1}\right\|_\Theta^2\nonumber\\
	&\le(\eta-1)\left\|\omega_t-\omega_{t+1}\right\|_\Theta^2+2\eta\left\|\omega_t-\omega_{t-1}\right\|_\Theta^2
	\end{align}
	
	Define a sequence $ s_k= \left\|\omega_{t}-\omega^*\right\|_\Theta^2-\eta\left\|\omega_{t-1}-\omega^*\right\|_\Theta^2 + 2\eta\left\|\omega_t-\omega_{t-1}\right\|_\Theta^2 $. Then
	\begin{align}\label{sequence_s}
	s_{k+1}-s_k&=\left\|\omega_{t+1}-\omega^*\right\|_\Theta^2 - \left\|\omega_{t}-\omega^*\right\|_\Theta^2 \nonumber\\
	&\quad -\eta(\left\|\omega_{t}-\omega^*\right\|_\Theta^2-\left\|\omega_{t-1}-\omega^*\right\|_\Theta^2) \nonumber\\
	&\quad + 2\eta\left\|\omega_{t+1}-\omega_{t}\right\|_\Theta^2 - 2\eta\left\|\omega_{t}-\omega_{t-1}\right\|_\Theta^2 \nonumber\\
	&\le (3\eta-1)\left\|\omega_{t+1}-\omega_{t}\right\|_\Theta^2
	\end{align}
	
	From \eqref{sequence_s}, we have $  (1-3\eta)\sum\nolimits_{t=1}^k\left\|\omega_{t+1}-\omega_{t}\right\|_\Theta^2\le s_1-s_{k+1}\le s_1+ \eta\left\|\omega_{t}-\omega^*\right\|_\Theta^2 $. 
	As $ 0<\eta<\frac{1}{3} $, we have $ s_{k+1}\le s_k\le s_1 $ and $ \left\|\omega_{t}-\omega^*\right\|_\Theta^2-\eta\left\|\omega_{t-1}-\omega^*\right\|_\Theta^2\le s_k\le s_1 $. Then, $ \left\|\omega_{t}-\omega^*\right\|_\Theta^2\le \eta^k(\left\|\omega_{1}-\omega^*\right\|_\Theta^2-\frac{\mu_1}{1-\alpha})+\frac{\mu_1}{1-\alpha} $, i.e., $\omega_{t}$ is bounded. Thus,  $  (1-3\eta)\sum\nolimits_{t=1}^k\left\|\omega_{t+1}-\omega_{t}\right\|_\Theta^2\le s_1+ \eta^{k+1}(\left\|\omega_{1}-\omega^*\right\|_\Theta^2-\frac{\mu_1}{1-\alpha})+\frac{\eta\mu_1}{1-\alpha} $. As $0<\eta<\frac{1}{3} $, we have 
	\begin{equation}\label{part1}
	\sum\nolimits_{t=1}^\infty\left\|\omega_{t+1}-\omega_{t}\right\|_\Theta^2<\infty
	\end{equation}
	Thus, $ \lim\limits_{t\to\infty}(\omega_{t+1}-\omega_{t})=0 $.
	
	Then, we prove that $ \left\|\omega_{t}-\omega^*\right\|_\Theta^2 $ also converges by the similar analysis in \cite[Theorem 6.3]{yi2017distributed}. We also write it here for completeness.
	
	Denote $\varphi_{t}=\max \left\{0,\left\|\omega_t-\omega^*\right\|_\Theta^2-\left\|\omega_{t+1}-\omega^*\right\|_\Theta^2\right\}$ and $\zeta_{t}=$
	$2\eta\left\|\omega_t-\omega_{t-1}\right\|_{\Theta}^{2}$. We know that $ \varphi_{t} $ is lower bounded. Recall \eqref{recursion}, and we have $\varphi_{t+1} \leq \eta \varphi_{t}+\zeta_{t}$. Apply this relationship recursively, and we have 
	\begin{equation}\label{middle3}
	\varphi_{t+1} \leq \eta^{t} \varphi_{1}+\sum\nolimits_{i=0}^{t-1} \eta^{i} \zeta_{t-i}
	\end{equation}
	Adding both sides of \eqref{middle3} from $ t=1\rightarrow\infty $, we have 
	\begin{equation}
	\sum\nolimits_{i=1}^{\infty} \varphi_{i}  \leq \frac{\varphi_{1}}{1-\eta}+\frac{1}{1-\eta} \sum\nolimits_{t=1}^{\infty} \zeta_{t} 
	\end{equation}
	From \eqref{part1}, we know $ \sum\nolimits_{t=1}^{\infty} \zeta_{t} <\infty $. This implies that $ \sum\nolimits_{t=1}^{\infty} \varphi_{t} $ is bounded, non-decreasing and  converges. 
	
	Consider another sequence $\left\{\left\|\omega_{t}-\omega^{*}\right\|_{\Theta}^{2}-\sum_{i=1}^{t} \varphi_{i}\right\} $, which is lower bounded. Moreover,
	\begin{equation}
	\setlength{\abovedisplayskip}{4pt}	
	\setlength{\belowdisplayskip}{4pt}
	\begin{split} 
	&\left\|\omega_{t+1}-\omega^{*}\right\|_{\Theta}^{2}-\sum\limits_{i=1}^{k+1} \varphi_{i} =\left\|\omega_{t+1}-\omega^{*}\right\|_{\Theta}^{2}-\varphi_{k+1}-\sum\limits_{i=1}^{t} \varphi_{i} \nonumber\\
	&\leq \left\|\omega_{t+1}-\omega^{*}\right\|_{\Theta}^{2}-\left\|\omega_{t+1}-\omega^{*}\right\|_{\Theta}^{2}+\left\|\omega_{t}-\omega^{*}\right\|_{\Theta}^{2}-\sum\limits_{i=1}^{t} \varphi_{i} \\
	&=\left\|\omega_{t}-\omega^{*}\right\|_{\Theta}^{2}-\sum\nolimits_{i=1}^{t} \varphi_{i} 
	\end{split}
	\end{equation}
	This implies that $\left\{\left\|\omega_{t}-\omega^{*}\right\|_{\Theta}^{2}-\sum_{i=1}^{t} \varphi_{i}\right\} $ is a non-increasing sequence and also converges. $ \left\|\omega_{t}-\omega^{*}\right\|_{\Theta}^{2} $ is the sum of two convergent sequences, and also converges.

	Since $\omega_{t}$ is bounded, it has a convergent subsequence $\omega_{n_t}$ converging to a point $\tilde\omega^*$. From $ \lim\limits_{t\to\infty}(\omega_{t+1}-\omega_{t})=0 $, we have $ \lim\limits_{t\to\infty}(\omega_{n_t+1}-\omega_{n_t})=0 $ and $ \lim\limits_{t\to\infty}(\omega_{n_t-1}-\omega_{n_t})=0 $. Due to the continuity of the righthand side of \eqref{compact_f}, we have $ \tilde\omega^*=(\rm Id+\Theta^{-1}\mathcal{U})^{-1}\tilde\omega^* $. Thus, $ \omega^* $ is an equilibrium point of the sequence $\omega_t$. This also implies that $ \left\|\omega_{t}-\tilde\omega^*\right\|_\Theta^2 $ converges. Because $ \left\|\omega_{n_t}-\tilde\omega^*\right\|_\Theta^2 $ converges to zero, $ \left\|\omega_{t}-\tilde\omega^*\right\|_\Theta^2 $ also converges to zero. 
	
	Based on Theorem \ref{Equivalence}, we have $ (p^*, b^*) $ is the GNE of \eqref{Sharing_problem_II}.
	This completes the proof.
\end{proof}



Invoking Theorem \ref{uniqueness}, if Assumption \ref{uniqueness} also holds, the algorithm will converge to the unique GNE of the generalized Nash game \eqref{Sharing_problem_II}, equivalently the energy sharing game  \eqref{Sharing_problem}.

\begin{figure}[t]
	\centering
	\includegraphics[width=0.47\textwidth]{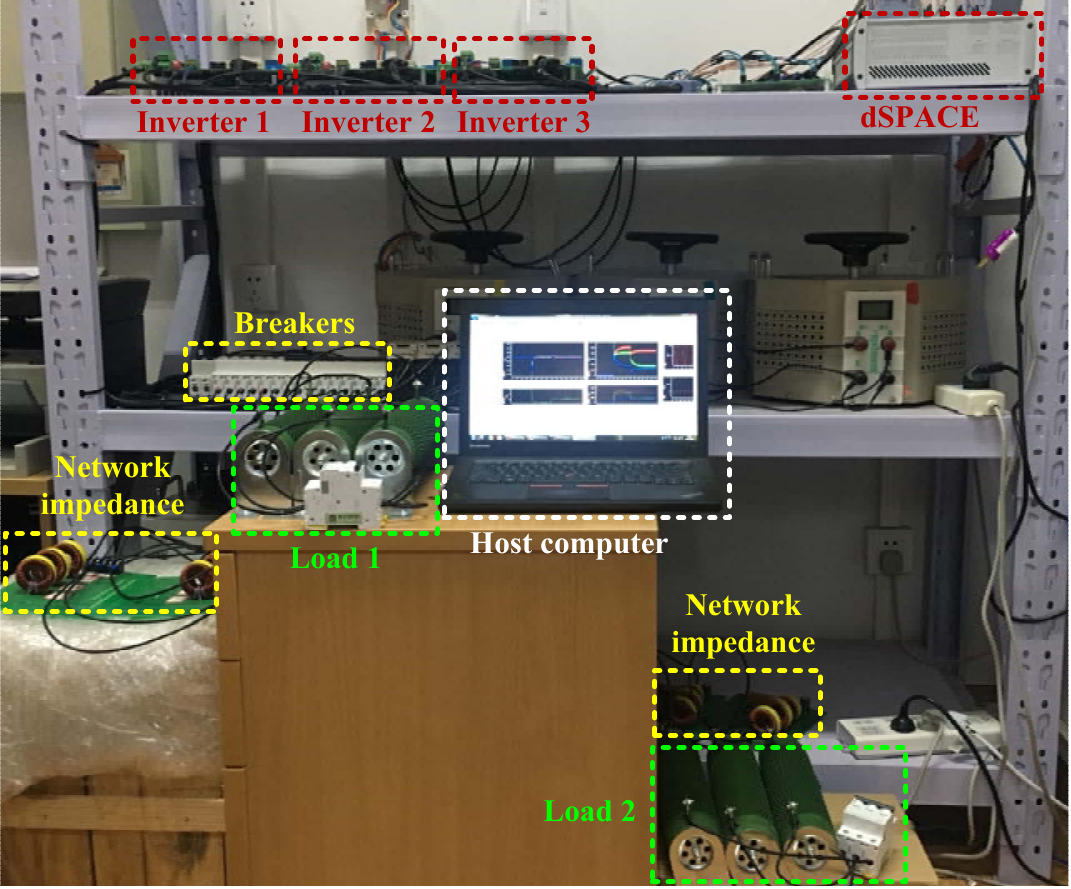}
	\caption{Experiment platform based on dSPACE RTI 1202 controller}
	\label{fig:experiment_platform}
\end{figure}

\begin{figure}[t]
	\centering
	\includegraphics[width=0.3\textwidth]{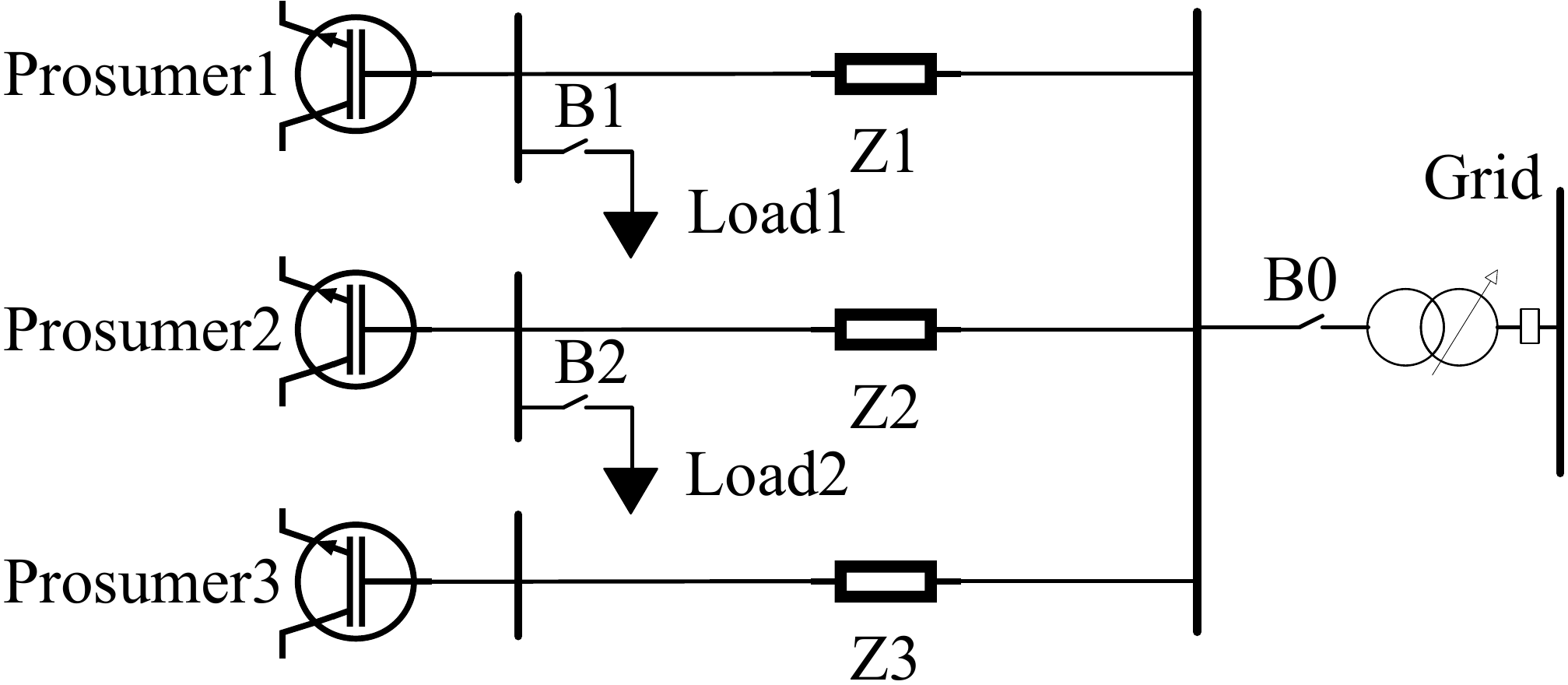}
	\caption{Topology of the experiment system}
	\label{fig:experiment_topology}
\end{figure}

\section{Experiments and Simulation Studies}\label{Experiments}
In this section, experiments and numerical simulations are introduced to verify the effectiveness of the proposed method. First, experiments with three prosumers are carried out to illustrate the basic properties of the algorithm. Then, a case with 123 prosumers is investigated to test the scalability, where the communication topology is identical to the topology of the IEEE 123-bus system\cite{kersting1991radial}. 

\subsection{Experimental results}

The proposed method is verified on an experimental platform based on the dSPACE RTI 1202 controller, which is presented in Fig.\ref{fig:experiment_platform}. It is composed of three inverters, one dSPACE RTI 1202 controller, two switchable loads, and one host computer. Each inverter represents a prosumer, which can both produce and consume power. The system topology is given in Fig.\ref{fig:experiment_topology}. In the experiments, the breaker B0 is open, i.e., the system operates in an isolated mode. One load is connected at the bus of Prosumer1 and the other is at the bus of Prosumer2. Three prosumers are connected through impedances. 
The communication topology is $\text{Prosumer1}\leftrightarrow \text{Prosumer2}\leftrightarrow \text{Prosumer3}\leftrightarrow \text{Prosumer1} $. 
The disutility function is $ h_i(p_i)=\frac{1}{2}c_ip_i^2+d_ip_i $ with $ c_1=0.00075, c_2=0.0006, c_3= 0.001, d_i=0 $. The price elasticity of each prosumer is set as $ a_i=-1000 $. The load demand of each prosumer is $D_1=730 \text{W}, D_2= 365 \text{W}, D_3=0$.

The simulation scenario is: 1) at $t=10$s, two loads are connected; 2) at $t=30$s, load 2 is disconnected. Then, each DG regulates its generation to balance the power difference. The frequency dynamics are illustrated in Fig.\ref{experiment_frequency}. When loads are connected, the frequency drops to about $ 49.2 $Hz and recovers to the nominal value in four seconds. On the contrary, when load 2 is switched off, the frequency increases and recovers in $ 3 $ seconds. 
\begin{figure}[t]
	\centering
	\includegraphics[width=0.32\textwidth]{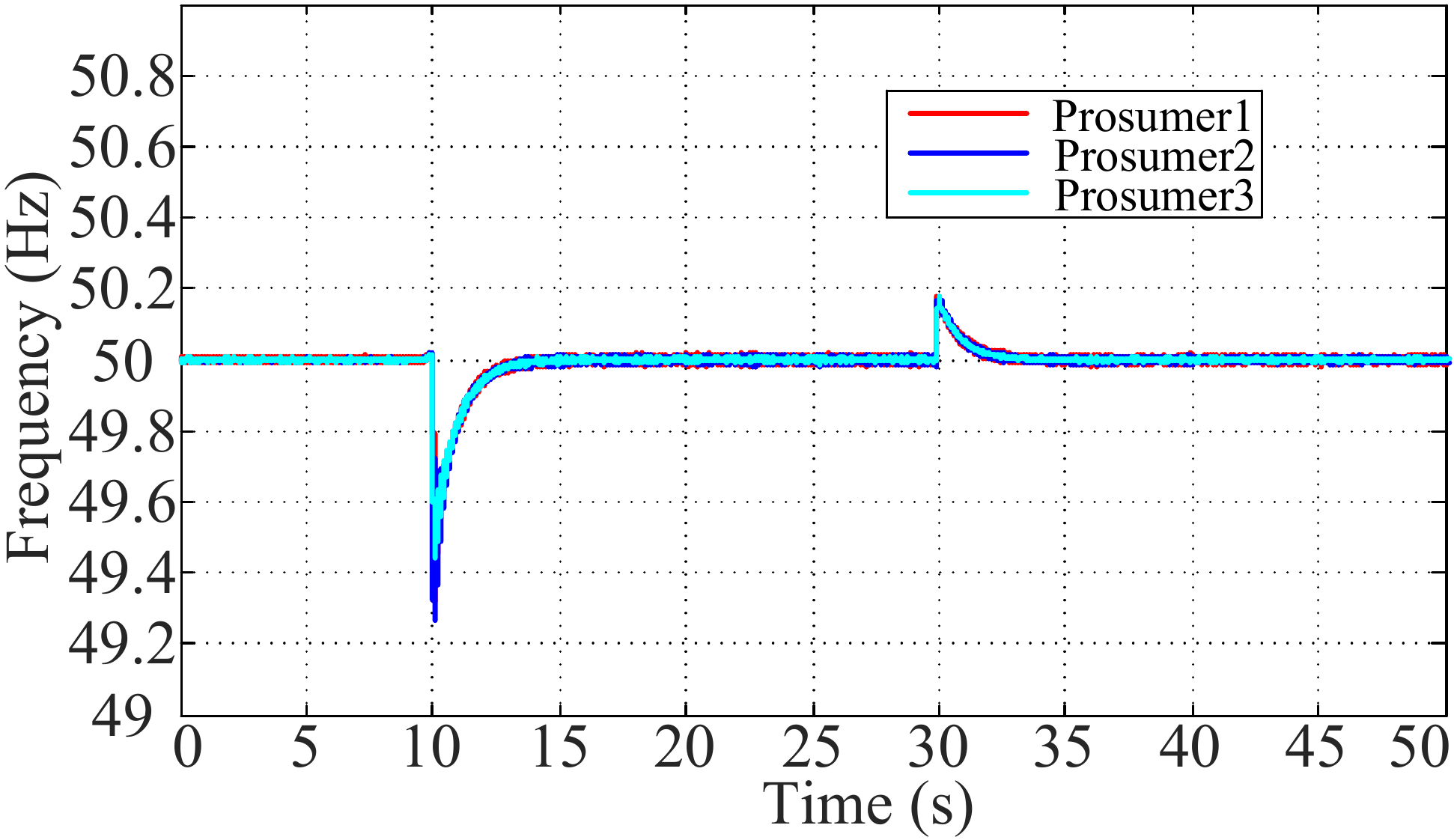}
	\caption{The frequency dynamics of three DGs }
	\label{experiment_frequency}
\end{figure}

\begin{figure}[t]
	\centering
	\includegraphics[width=0.475\textwidth]{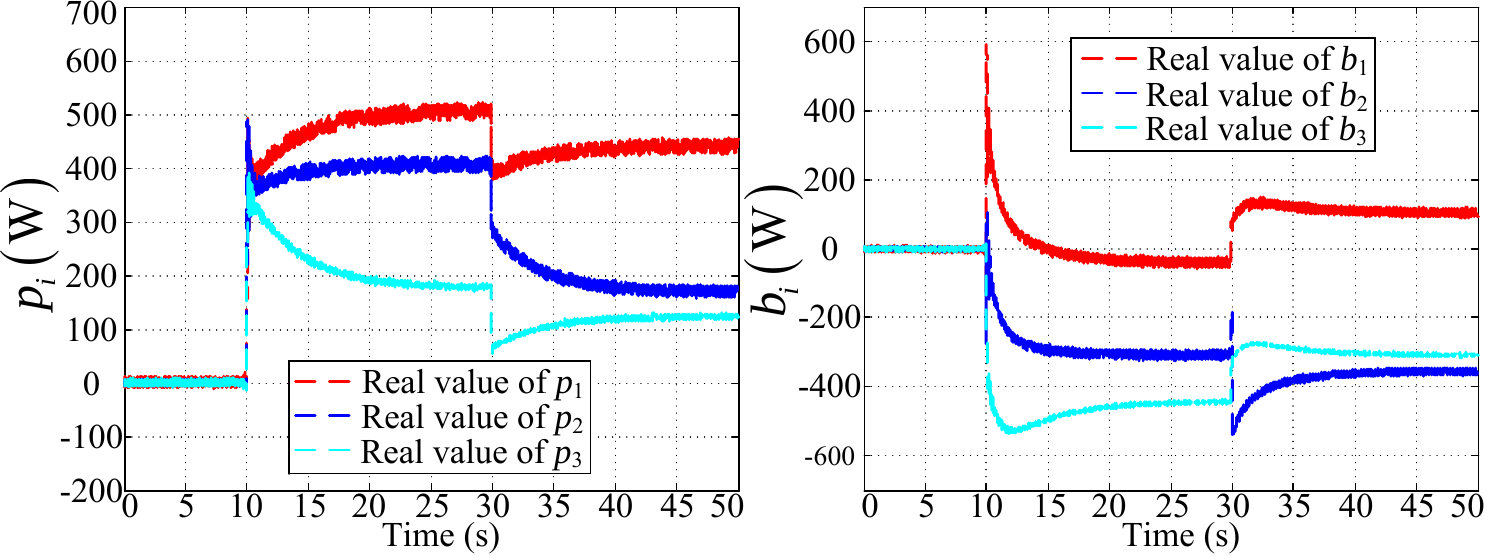}
	\caption{Dynamics of the GNE}
	\label{experiment_power}
\end{figure}
The GNE of the first stage is $ (p_1^*=505.4 \text{W}, b_1^*=103.3\text{W} ), (p_2^*=408.4\text{W}, b_2^*=355.0\text{W} ), (p_3^*=177.8\text{W}, b_3^*=310.0\text{W} )$. 
Dynamics of seeking the GNE is illustrated in Fig.\ref{experiment_power}, where the left part is the generation $p_i$ of each prosumer and the right part is the purchase willingness $ b_i $. They vary slightly around the equilibrium. In this stage, Prosumer1 buys power from Prosumer2 and Prosumer3. 
The GNE of the second stage is $ (p_1^*=440.6 \text{W}, b_1^*=-41.7\text{W} ), (p_2^*=168.9\text{W}, b_2^*=-310.0\text{W} ), (p_3^*=123.9\text{W}, b_3^*=-444.2\text{W} )$.
Compared with the results obtained from the centralized method, the steady state generations are optimal to the problem \eqref{Central_problem_II}. This shows that the GNE is obtained. Moreover, real power can always trace the reference value. This verifies that the proposed method can get the correct results in the experiment. 
\begin{figure}[t]
	\centering
	\includegraphics[width=0.42\textwidth]{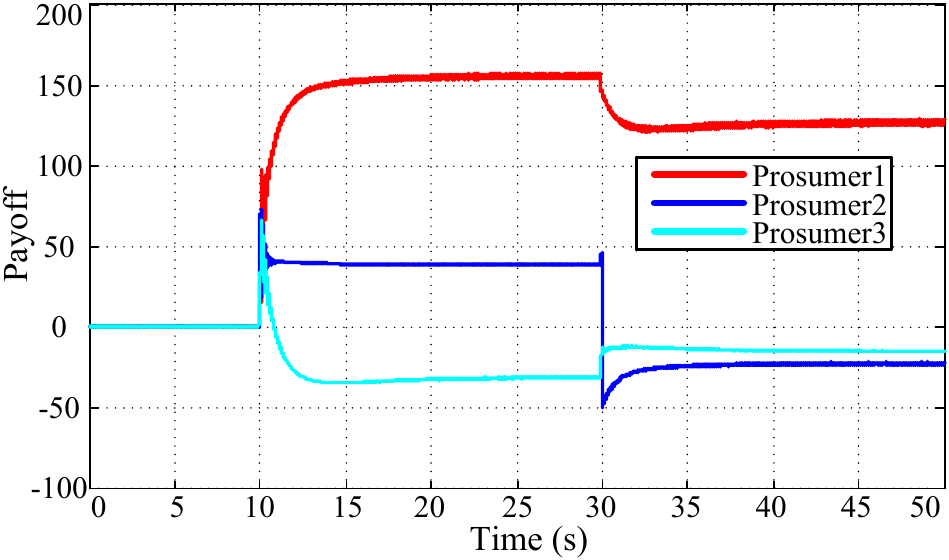}
	\caption{Dynamics of the payoff functions }
	\label{experiment_payoff}
\end{figure}
Dynamics of the payoff functions are given in Fig.\ref{experiment_payoff}, which also converges in five seconds. In the first stage, Prosumer3 earns profit by selling power to Prosumer1 and Prosumer2. In the second stage, there is only Prosumer1 buying power, while others selling. The cost of Prosumer2 changes from positive to negative, which implies that it earns profit by selling power to Prosumer1. 

\subsection{Scalability}
In this subsection, the scalability of proposed method is illustrated. The 123 prosumers is numbered with the same number of IEEE 123-bus system. $ c_i $ in the disutility function is $ c_i=1+\frac{i}{123} $ and $ d_i=1 $.  $ a_i $ of these 123 prosumers are all $ -0.1 $. 
The effect of $\eta$ is also investigated, which is set as $ 0, 0.1, 0.2, 0.33 $ respectively. The convergence of $ p_{i,t} $ to the equilibrium point is given in Fig.\ref{different_eta}. 
\begin{figure}[t]
	\centering
	\includegraphics[width=0.47\textwidth]{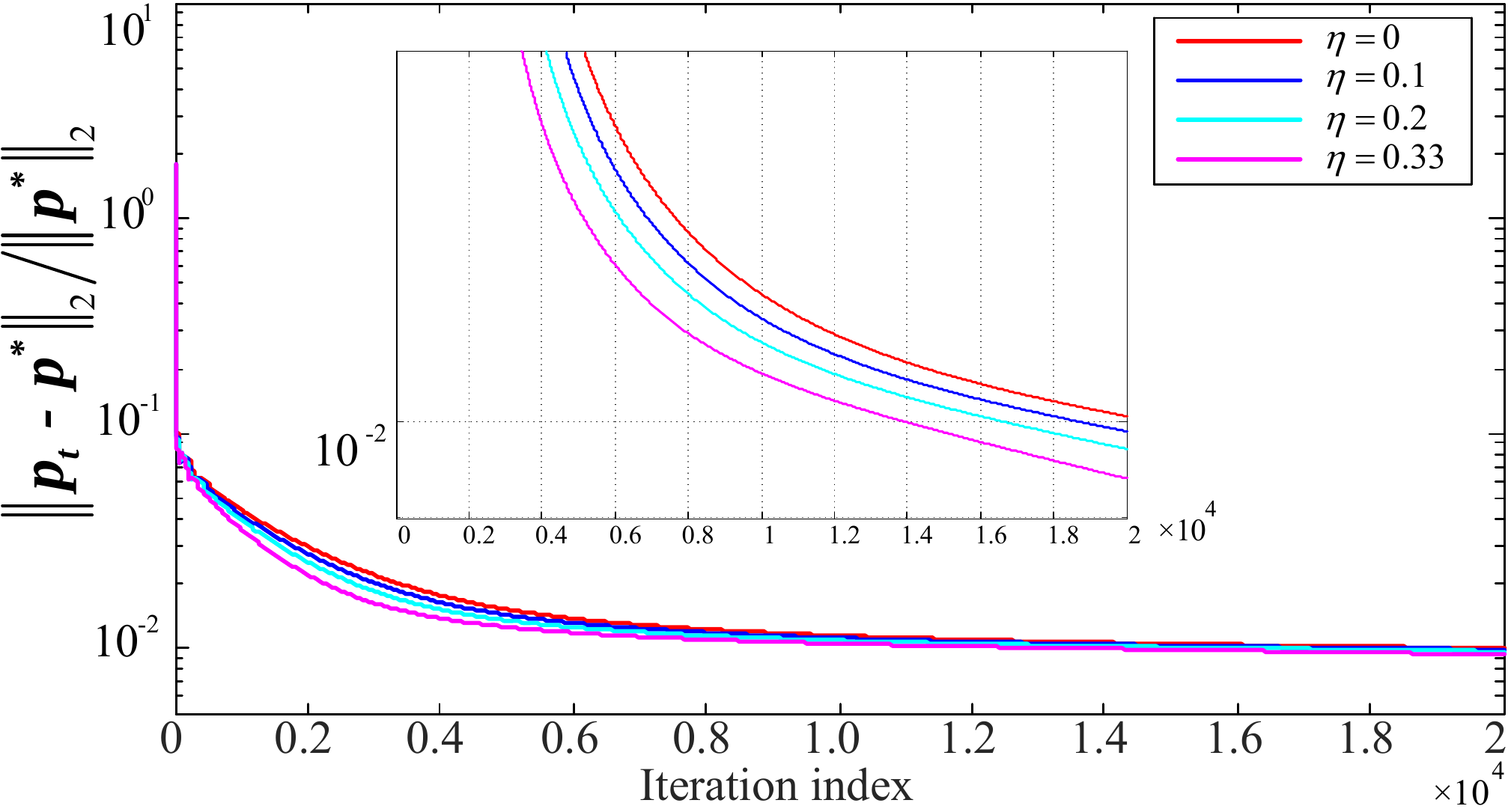}
	\caption{Error of 123 prosumers with different $\eta$}
	\label{different_eta}
\end{figure}
With $ \eta=0 $, it takes more than $ 20,000 $ iterations to make the relative error smaller than $ 1\% $. With $ \eta=0.33 $, it takes about $ 14,000 $ iterations, which is reduced by $ 30\% $. It is also observed that a larger $ \eta $ results in a higher convergence rate.

\section{Conclusion}\label{Conclusion}
In this work, the energy sharing game among prosumers is formulated, whose objective function for each prosumer is associated with decisions of all the prosumers. To solve it in a fully distributed way, the game is transformed to an equivalent optimization problem. Then, a Krasnosel'ski{\v{i}}-Mann iteration type algorithm is devised to solve the problem and consequently find the GNE in a distributed manner  with neighboring communication. The convergence of the algorithm is proved rigorously. Experimental results with three prosumers show that the GNE can be obtained at a fast speed. Simulations on 123 prosumers verify the scalability of the proposed method.

This paper offers a different perspective for seeking the GNE in the GNG. If the objective function of one prosumer is associated with too many other prosumers, it is possible to transform the GNG to an equivalent optimization problem to avoid dense communication. In our future work, it is interesting to investigate the conditions of the equivalent transformation for a more general GNG. 

\bibliographystyle{IEEEtran}
\bibliography{mybib}

\appendices
\makeatletter
\@addtoreset{equation}{section}
\@addtoreset{theorem}{section}
\makeatother
\renewcommand{\theequation}{A.\arabic{equation}}
\renewcommand{\thetheorem}{A.\arabic{theorem}}
\setcounter{equation}{0}

\section{Proofs of Theorem \ref{unique}}\label{Appendix_Unique}
\begin{proof}
	$\Rightarrow$ 1) From Lemma \ref{lemma_VI}, we need to find $x^*:=(p^*_i, b^*_i) \in X$ such that 
	\begin{equation}\label{variation_inequality}
	\left\langle F\left({x}^{*}\right), {x}-{x}^{*}\right\rangle \geq 0, \forall {x} \in X
	\end{equation}
	
	Now, we check the KKT condition for $VI(X, {F}(p, b))$ in \eqref{variation_inequality}. In fact, ${x}^{*}$ is a solution to $VI(X, {F}(p, b))$ in \eqref{variation_inequality} if and only if ${x}^{*}$ is the optimal solution to the following optimization problem:
	\begin{equation}\label{variation_inequality2}
	\min _{{x} \in \mathbb{R}^{2n}}\left\langle F\left({x}^{*}\right), {x}\right\rangle, \quad \text { s.t. }\ {x} \in X
	\end{equation}
	If ${x}^{*}$ solves \eqref{variation_inequality2}, there exists $\lambda^{*} \in \mathbb{R}$ such that the
	following optimality conditions (KKT) are satisfied \cite[Theorem 3.25]{ruszczynski2006nonlinear}
	\begin{subequations}
		\begin{align}
		\label{KKT_game1}
		{0} &\in h^{'}_i(p_i^{*}) + \lambda_i^{*}+N_{\Omega_{i}}\left(p_i^{*}\right) \\
		\label{KKT_game2}
		0&= - \left( {2a_i{\mu _c} + {b_i^{*}} + a_i{\lambda _i^{*}} } \right)\frac{1}{{\textbf{1}^{\rm T}a}} + {\mu _c} + {\lambda ^{*}_i}\\
		\label{KKT_game3}
		0&=p_i^{*}+a_i\mu_c^*+b_i^{*}-D_i \\
		\label{KKT_game4}
		0&=\textbf{1}^{\rm T}a \mu_c^*+b_i^{*}+\sum\nolimits_{j\neq i}b_j^{*}
		\end{align}
		\label{KKT_game}From \eqref{KKT_game2}, we know 
	\end{subequations}
	\begin{equation}
	0= - \left( {2a_i{\mu _c^*} + {b_i^{*}} + a_i{\lambda _i^{*}} } \right)\frac{1}{{\textbf{1}^{\rm T}a}} + {\mu _c^*} + {\lambda _i^{*}} \nonumber\\
	\end{equation}
	\begin{equation}\label{EQlambda}
	\setlength{\abovedisplayskip}{4pt}	
	\setlength{\belowdisplayskip}{4pt}
	\Rightarrow \lambda^{*}_i=\frac{\textbf{1}^{\rm T}a-2a_i}{a_i-\textbf{1}^{\rm T}a}\mu^*_c-\frac{1}{a_i-\textbf{1}^{\rm T}a}b_i^{*}
	\end{equation}
	From \eqref{KKT_game3}, we know
	\begin{equation}\label{EQb}
	b_i^{*}=D_i-p^{*}_i-a_i\mu_c^*
	\end{equation}
	Combining \eqref{EQlambda} and \eqref{EQb}, we have 
	\begin{align}\label{EQlambda2}
	\lambda_i^{*}&=\frac{\textbf{1}^{\rm T}a-2a_i}{a_i-\textbf{1}^{\rm T}a}\mu_c^*-\frac{D_i-p_i^{*}-a_i\mu_c^*}{a_i-\textbf{1}^{\rm T}a}\nonumber\\
	&=-\mu_c^* + \frac{p^{*}_i}{a_i-\textbf{1}^{\rm T}a}-\frac{D_i}{a_i-\textbf{1}^{\rm T}a}
	\end{align}
	Then, the KKT condition \eqref{KKT_game} is 
	\begin{subequations}
		\label{KKT_game_eq}
		\begin{align}\label{KKT_game1e}
		{0} &\in h^{'}_i(p_i^{*}) +\frac{p_i^{*}}{a_i-\textbf{1}^{\rm T}a} -\frac{D_i}{a_i-\textbf{1}^{\rm T}a} -\mu_c^{*}+N_{\Omega_{i}}\left(p_i^{*}\right)\\
		0 &=\sum\nolimits_{i\in\mathcal{N}} D_i -\sum\nolimits_{i\in\mathcal{N}} p_i
		\end{align}
	\end{subequations}
	It is also the KKT condition of the following problem
	\begin{subequations}
		\label{Central_problem_III}     
		\setlength{\abovedisplayskip}{4pt}	
		\setlength{\belowdisplayskip}{4pt}
		\begin{align}
		\min\limits_{p}\ & \hat h(p)=\sum\limits_{i\in\mathcal{N}} \left(h_i(p_i)+ \frac{p_i^2}{2(a_i-\textbf{1}^{\rm T}a)}-\frac{D_i}{a_i-\textbf{1}^{\rm T}a}p_i\right)
		\label{Central_problem1_III}
		\\ 
		\text{s.t.} \quad
		&  \sum\nolimits_{i\in\mathcal{N}} p_i =\sum\nolimits_{i\in\mathcal{N}} D_i,\quad \mu_c
		\label{Central_problem2_III}
		\\
		\label{Central_problem3_III}
		&\underline p_i\le p_i\le \overline p_i  
		\end{align}
	\end{subequations}
	where $\mu_c$ is the Lagrangian multiplier. 
	
	Define 
	\begin{align}		
	\tilde{h}_i(p_i)= h_i(p_i)+ \frac{p_i^2}{2(a_i-\textbf{1}^{\rm T}a)}-\frac{D_i}{a_i-\textbf{1}^{\rm T}a}p_i 
	\end{align} 
	We know $  \tilde{h}_i(p_i) $ is strongly convex and the Slater's condition holds by Assumption \ref{Slater}. Thus, $ p_i^* $ and $ \mu_c^* $ exist. By \eqref{EQb}, $ b_i^* $ also exists.
	
	$\Rightarrow$ 2) If Assumption \ref{uniqueness} holds, there exists at least one prosumer $ i $ with $ \underline p_i< p_i^{*}< \overline p_i $. For this $ i $, we have $ \left\{0\right\}=N_{\Omega_{i}}\left(p_i^{*}\right) $. Then, by \eqref{KKT_game1e}, $ \mu_c^* $ is unique. The unique $ b_i^* $ can be obtained from \eqref{EQb}. 
\end{proof}

\ifCLASSOPTIONcaptionsoff
  \newpage
\fi

\end{document}